\documentclass{sig-alternate-05-2015}
\usepackage{tabularx,colortbl}
\usepackage[dvipsnames]{xcolor}
\usepackage{flushend}
\usepackage{cite}
\usepackage{amsmath}
\usepackage{amssymb}
\usepackage{epsfig}
\usepackage{stmaryrd}
\usepackage{url}
\usepackage{multirow}
\usepackage{latexsym}
\usepackage{graphics}
\usepackage{graphicx}
\usepackage{enumitem}
\usepackage{comment}
\usepackage{longtable}
\usepackage{supertabular}
\usepackage{times}
\usepackage{listings}
\usepackage{subfigure}
\usepackage{color}
\usepackage{balance}
\usepackage[ruled, vlined, linesnumbered]{algorithm2e}

\sloppypar

\begin{document}
\CopyrightYear{2016}
\setcopyright{acmcopyright}
\conferenceinfo{FSE'16,}{November 13-19, 2016, Seattle, WA, USA}
\isbn{978-1-4503-4218-6/16/11}\acmPrice{\$15.00}
\doi{http://dx.doi.org/10.1145/2950290.2950346}
\definecolor{gold}{rgb}{0.90,.66,0}
\definecolor{dgreen}{rgb}{0,0.6,0}
\newcommand{\mike}[1]{\textcolor{red}{#1}}
\newcommand{\fixed}[1]{\textcolor{purple}{#1}}
\newcommand{\andrew}[1]{\textcolor{green}{#1}}
\newcommand{\ela}[1]{\textcolor{blue}{#1}}
\newcommand{\stateequiv}{\equiv_{s}}
\newcommand{\traceequiv}{\equiv_{\sigma}}
\newcommand{\ta}{\text{TA}}
\newcommand{\cta}{\text{TA$_{C}$}}
\newcommand{\tta}{\text{TA$_{T}$}}

\newdef{lemma}{Lemma}
\newdef{definition}{Definition}
\newdef{theorem}{Theorem}
\newdef{note}{Note}

\title{Efficient Generation of Inductive Validity Cores for Safety Properties}

\numberofauthors{3}

\author{
\alignauthor
Elaheh Ghassabani\\
       \affaddr{Department of Computer Science and Engineering}\\
       \affaddr{University of Minnesota}\\
       \affaddr{\small{200 Union Street\\Minneapolis, MN, 55455, USA}}\\
       \email{\texttt{\small{ghass013@umn.edu}}}
\alignauthor
Andrew Gacek\\
       \affaddr{Rockwell Collins\\ Advanced Technology Center}\\
       \affaddr{\small{400 Collins Rd. NE\\Cedar Rapids, IA, 52498, USA}}\\
       \email{\texttt{\small{andrew.gacek@rockwellcollins.com}}}
\alignauthor
Michael W. Whalen\\
       \affaddr{Department of Computer Science and Engineering}\\
       \affaddr{University of Minnesota}\\
       \affaddr{\small{200 Union Street\\Minneapolis, MN, 55455, USA}}\\
       \email{\texttt{\small{whalen@cs.umn.edu}}}
}

\maketitle

\begin{abstract}
Symbolic model checkers can construct proofs of properties over very complex models.  However, the results reported by the tool when a proof succeeds do not generally provide much insight to the user.  It is often useful for users to have traceability information related to the proof: which portions of the model were necessary to construct it.  This traceability information can be used to diagnose a variety of modeling problems such as overconstrained axioms and underconstrained properties, and can also be used to measure {\em completeness} of a set of requirements over a model.  In this paper, we present a new algorithm to efficiently compute the {\em inductive validity core} (IVC) within a model necessary for inductive proofs of safety properties for sequential systems.  The algorithm is based on the UNSAT core support built into current SMT solvers and a novel encoding of the inductive problem to try to generate a minimal inductive validity core.  We prove our algorithm is correct, and describe its implementation in the JKind model checker for Lustre models.  We then present an experiment in which we benchmark the algorithm in terms of speed, diversity of produced cores, and minimality, with promising results.
\end{abstract}



\section{Introduction}
\label{sec:intro}

Symbolic model checking using induction-based techniques such as IC3/PDR~\cite{Een2011:PDR} and $k$-induction~\cite{SheeranSS00} can often determine whether safety properties hold of complex finite or infinite-state systems.  Model checking tools are attractive both because they are automated, requiring little or no interaction with the user, and if the answer to a correctness query is negative, they provide a counterexample to the satisfaction of the property.  These counterexamples can be used both to illustrate subtle errors in complex hardware and software designs~\cite{hilt2013,McMillan99:compositional, Miller10:CACM} and to support automated test case generation~\cite{Whalen13:OMCDC, You15:dse}.
In the event that a property is proved, however, it is not always clear what level of assurance should be invested in the result.  Given that these kinds of analyses are performed for safety- and security-critical software, this can lead to overconfidence in the behavior of the fielded system.  It is well known that issues such as vacuity~\cite{Kupferman03:Vacuity} can cause verification to succeed despite errors in a property specification or in the model. Even for non-vacuous specifications, it is possible to over-constrain the specification of the {\em environment} in the model such that the implementation will not work in the actual operating environment.

At issue is the level of feedback provided by the tool to the user. In
most tools, when the answer to a correctness query is positive, no
further information is provided. What we would like to provide is
traceability information, an {\em inductive validity core} (IVC), that explains
the proof, in much the same way that a counterexample explains the
negative result. This is not a new idea: UNSAT cores~\cite{zhang2003extracting}
provide the same kind of information for individual SAT or
SMT queries, and this approach has been lifted to bounded analysis
for Alloy in~\cite{Torlak08:cores}. What we propose is a generic and efficient
mechanism for extracting supporting information, similar to an UNSAT
core, from the proofs of safety properties using inductive techniques
such as PDR and $k$-induction. Because many
properties are not themselves inductive, these proof techniques
introduce lemmas as part of the solving process in order to strengthen
the properties and make them inductive. Our technique allows
efficient, accurate, and precise extraction of inductive validity cores
even in the presence of such auxiliary lemmas.

Once generated, the IVC can be used for many purposes in the software verification process, including at least the following:
\begin{description}
    \item[Vacuity detection:] The idea of syntactic vacuity detection (checking whether all subformulae within a property are necessary for its satisfaction) has been well studied~\cite{Kupferman03:Vacuity}.   However, even if a property is not syntactically vacuous, it may not require substantial portions of the model.  This in turn may indicate that either a.) the model is incorrectly constructed or b.) the property is weaker than expected.  We have seen several examples of this mis-specification in our verification work, especially when variables computed by the model are used as part of antecedents to implications.
    \item[Completeness checking:] Closely related to vacuity detection is the idea of {\em completeness checking}, e.g., are all atoms in the model necessary for at least one of the properties proven about the model?  Several different notions of completeness checking have been proposed~\cite{chockler_coverage_2003, kupferman_theory_2008}, but these are very expensive to compute, and in some cases, provide an overly strict answer (e.g., checking can only be performed on non-vacuous models for~\cite{kupferman_theory_2008}). 
    \item[Traceability:] Certification standards for safety-critical systems (e.g.,~\cite{DO178C, MOD:00-55}) usually require {\em traceability matrices} that map high-level requirements to lower-level requirements and (eventually) leaf-level requirements to code or models.  Current traceability approaches involve either manual mappings between requirements and code/models~\cite{SimulinkTraceability} or a heuristic approach involving natural language processing~\cite{Keenan12:Tracelab}.  Both of these approaches tend to be inaccurate.  For functional properties that can be proven with inductive model checkers, inductive validity cores can provide accurate traceability matrices with no user effort.
    \item[Symbolic Simulation / Test Case Generation:]  Model checkers are now often used for symbolic simulation and structural-coverage-based test case generation~\cite{SimulinkDesignVerifier,Whalen13:OMCDC}.  For either of these purposes, the model checker is supposed to produce a witness trace for a given coverage obligation using a ``trap property'' which is expected to be falsifiable.  In systems of sufficient size, there is often ``dead code'' that cannot ever be reached.  In this case, a proof of non-reachability is produced, and the IVC provides the reason why this code is unreachable.
\end{description}
\noindent Nevertheless, to be useful for these tasks, the generation
process must be efficient and the generated IVC must be
accurate and precise (that is, sound and close to minimal).  The requirement for accuracy is obvious; otherwise the ``minimal'' set of model elements is no longer sufficient to produce a proof, so it no longer meets our IVC definition.  Minimality is important because (for traceability) we do not want unnecessary model elements in the trace matrix, and (for completeness) it may give us a false level of confidence that we have enough requirements.

In addition, 
we are also interested in {\em diversity}:  how many different IVCs can be computed for a given property and model? Requirements engineers often talk about ``the traceability matrix'' or ``the satisfaction argument''.  If proofs are regularly diverse, then there are potentially many equally valid traceability matrices, and this may lead to changes in traceability research.

\begin{figure}[t]
\centering
\includegraphics[width=\columnwidth]{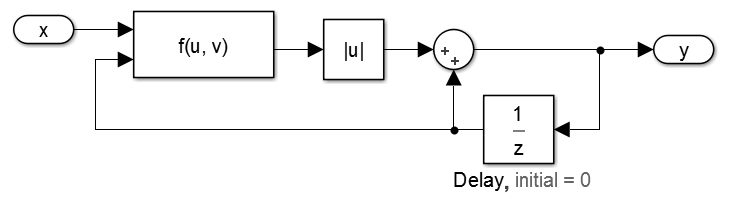}
{\smaller
\begin{verbatim}
node filter(x : real) returns (a, b, y : real);
let
  a = f(x, 0.0 -> pre y);
  b = if a >= 0.0 then a else -a;
  y = b + (0.0 -> pre y);
tel;
\end{verbatim}
}
\vspace{-0.1in}
\caption{Model with property $y \geq 0$, before IVC analysis}
\label{fig:ex-before}
\end{figure}

In the remainder of this paper, we present an algorithm for efficient generation of IVCs for induction-based model checkers.  Our contributions, as detailed in the remainder of the paper, are as follows:

\begin{itemize}
    \item We present a technique for extracting inductive validity
      cores from an inductive verification of a safety property over a sequential model involving lemmas.
    \item We formalize this technique and present an implementation of it in the JKind model checker~\cite{jkind}.
    \item We present an experiment over our implementation and measure the efficiency, minimality, and robustness of the IVC generation process.
\end{itemize}

The rest of this article is organized as follows. In
Section~\ref{sec:exmpl}, we present a motivating example. In
Section~\ref{sec:background}, we present the required background for
our approach. In Sections~\ref{sec:ivc} and~\ref{sec:impl}, we present
our approach and our implementation in JKind.
Sections~\ref{sec:experiment} and~\ref{sec:results} present an
evaluation of our approach on a set of benchmark examples. Finally,
Section~\ref{sec:related} discusses related work and
Section~\ref{sec:conc} concludes.

\section{Motivating Example}
\label{sec:exmpl}

\begin{figure}[t]
\includegraphics[width=\columnwidth]{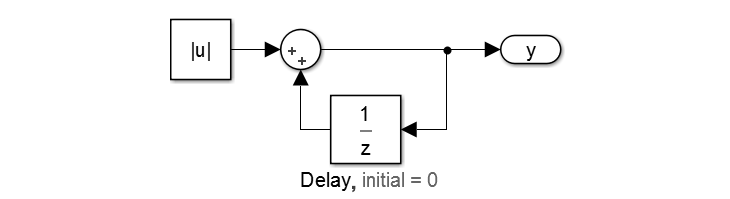}
{\smaller
\begin{verbatim}
node filter(x, a : real) returns (b, y : real);
let
  b = if a >= 0.0 then a else -a;
  y = b + (0.0 -> pre y);
tel;
\end{verbatim}
}
\vspace{-0.1in}
\caption{Model with property $y \geq 0$, after IVC analysis}
\label{fig:ex-after}
\end{figure}

Consider the model shown both graphically and textually in
Figure~\ref{fig:ex-before}. This model takes an input, combines it
with the previous output in some way, takes the absolute value, and
then adds this to an accumulating value. This model has the property
that the output is always non-negative, i.e., $y \geq 0$. Moreover, it
happens that this property holds regardless of the way the input is
combined with the previous output, i.e., the function $f$ in the
model. Formally, we say that the minimal inductive validity core (IVC) does
not contain that part of the model. The model reduced to a minimal IVC
is shown in Figure~\ref{fig:ex-after}. Note that traditional static
dependency analysis (i.e., a {\em backward static slice}) would not be able to
remove $f$ from the original model.  In our experiments in Section~\ref{sec:results},
we demonstrate that IVCs are much smaller and more precise than static slices.



\section{Preliminaries}
\label{sec:background}

\newcommand{\bool}[0]{\mathit{bool}}
\newcommand{\reach}[0]{\mathit{R}}
\newcommand{\ite}[3]{\mathit{if}\ {#1}\ \mathit{then}\ {#2}\ \mathit{else}\ {#3}}


Given a state space $S$, a transition system $(I,T)$ consists of an
initial state predicate $I : S \to \bool$ and a transition step
predicate $T : S \times S \to \bool$. We define the notion of
reachability for $(I, T)$ as the smallest predicate $\reach : S \to
\bool$ which satisfies the following formulas:
\begin{gather*}
  \forall s.~ I(s) \Rightarrow \reach(s) \\
  \forall s, s'.~ \reach(s) \land T(s, s') \Rightarrow \reach(s')
\end{gather*}
A safety property $P : S \to \bool$ is a state predicate. A safety
property $P$ holds on a transition system $(I, T)$ if it holds on all
reachable states, i.e., $\forall s.~ \reach(s) \Rightarrow P(s)$,
written as $\reach \Rightarrow P$ for short. When this is the case, we
write $(I, T)\vdash P$.

For an arbitrary transition system $(I, T)$, computing reachability
can be very expensive or even impossible. Thus, we need a more
effective way of checking if a safety property $P$ is satisfied by the
system. The key idea is to over-approximate reachability. If we can
find an over-approximation that implies the property, then the
property must hold. Otherwise, the approximation needs to be refined.

A good first approximation for reachability is the property itself.
That is, we can check if the following formulas hold:
\begin{gather}
  \forall s.~ I(s) \Rightarrow P(s)
  \label{eq:1-ind-base} \\
  \forall s, s'.~ P(s) \land T(s, s') \Rightarrow P(s')
  \label{eq:1-ind-step}
\end{gather}
If both formulas hold then $P$ is {\em inductive} and holds over the
system. If (\ref{eq:1-ind-base}) fails to hold, then $P$ is violated
by an initial state of the system. If (\ref{eq:1-ind-step}) fails to
hold, then $P$ is too much of an over-approximation and needs to be
refined.

One way to refine our over-approximation is to add additional lemmas
to the property of interest. For example, given another property $L :
S \to bool$ we can consider the extended property $P'(s) = P(s) \land
L(s)$, written as $P' = P \land L$ for short. If $P'$ holds on the
system, then $P$ must hold as well. The hope is that the addition of
$L$ makes formula (\ref{eq:1-ind-step}) provable because the
antecedent is more constrained. However, the consequent of
(\ref{eq:1-ind-step}) is also more constrained, so the lemma $L$ may
require additional lemmas of its own. Finding and proving these
lemmas is the means by which property directed reachability (PDR)
strengthens and proves a safety property.

Another way to refine our over-approximation is to use use {\em
  $k$-induction} which unrolls the property over $k$ steps of the
transition system. For example, 1-induction consists of formulas
(\ref{eq:1-ind-base}) and (\ref{eq:1-ind-step}) above, whereas
2-induction consists of the following formulas:
\begin{gather*}
\forall s.~ I(s) \Rightarrow P(s) \\
\forall s, s'.~ I(s) \land T(s, s') \Rightarrow P(s') \\
\forall s, s', s''.~ P(s) \land T(s, s') \land P(s') \land T(s',
  s'') \Rightarrow P(s'')
\end{gather*}
That is, there are two base step checks and one inductive step check.
In general, for an arbitrary $k$, $k$-induction consists of $k$
base step checks and one inductive step check as shown in
Figure~\ref{fig:k-induction} (the universal quantifiers on $s_i$ have
been elided for space). We say that a property is $k$-inductive if it
satisfies the $k$-induction constraints for the given value of $k$.
The hope is that the additional formulas in the antecedent of the
inductive step make it provable.

\begin{figure}
\begin{gather*}
I(s_0) \Rightarrow P(s_0) \\[-2pt]
\vdots \\[2pt]
I(s_0) \land T(s_0, s_1) \land \cdots \land T(s_{k-2}, s_{k-1})
\Rightarrow P(s_{k-1}) \\[2pt]
P(s_0) \land T(s_0, s_1) \land \cdots \land P(s_{k-1}) \land
T(s_{k-1}, s_k) \Rightarrow P(s_k)
\end{gather*}
\caption{$k$-induction formulas: $k$ base cases and one inductive
  step}
\label{fig:k-induction}
\end{figure}

In practice, inductive model checkers often use a combination of the
above techniques. Thus, a typical conclusion is of the form ``$P$ with
lemmas $L_1, \ldots, L_n$ is $k$-inductive''.



\section{Inductive Validity Cores}
\label{sec:ivc}

\newcommand{\bfalg}{IVC\_BF\xspace}
\newcommand{\ucalg}{IVC\_UC\xspace}
\newcommand{\ucbfalg}{IVC\_UCBF\xspace}
\newcommand{\bq}{\textsc{BaseQuery}\xspace}
\newcommand{\iq}{\textsc{IndQuery}\xspace}
\newcommand{\fq}{\textsc{FullQuery}\xspace}

\newcommand{\mink}{\textsc{MinimizeK}\xspace}
\newcommand{\reduceinv}{\textsc{ReduceInvariants}\xspace}
\newcommand{\minivc}{\textsc{MinimizeIvc}\xspace}

\newcommand{\checksat}{\textsc{CheckSat}\xspace}
\newcommand{\unsatcore}{\textsc{UnsatCore}\xspace}
\newcommand{\unsat}{\textsc{UNSAT}\xspace}
\newcommand{\sat}{\textsc{SAT}\xspace}

Given a transition system which satisfies a safety property $P$, we
want to know which parts of the system are necessary for satisfying
the safety property. One possible way of asking this is, ``What is the
most general version of this transition system that still satisfies
the property?'' The answer is disappointing. The most general system is
$I(s) = P(s)$ and $T(s, s') = P(s')$, i.e., you start in any state
satisfying the property and can transition to any state that still
satisfies the property. This answer gives no insight into the original
system because it has no connection to the original system. In this
section we introduce the notion of {\em inductive validity cores} (IVC)
which looks at generalizing the original transition system while
preserving a safety property.

In order to talk about generalizing a transition system, we assume the
transition relation of the system has the structure of a top-level
conjunction. This assumption gives us a structure that we can easily
manipulate as we generalize the system. Given $T(s, s') = T_1(s, s')
\land \cdots \land T_n(s, s')$ we will write $T = T_1 \land \cdots
\land T_n$ for short. By further abuse of notation we will identify
$T$ with the set of its top-level conjuncts. Thus we will write $x \in
T$ to mean that $x$ is a top-level conjunct of $T$. We will write $S
\subseteq T$ to mean that all top-level conjuncts of $S$ are top-level
conjuncts of $T$. We will write $T \setminus \{x\}$ to mean $T$
with the top-level conjunct $x$ removed. We will use the same notation
when working with sets of invariants.

\begin{definition}{\emph{Inductive Validity Core:}}
  \label{def:ivc}
  Let $(I, T)$ be a transition system and let $P$ be a
  safety property with $(I, T)\vdash P$. We say $S \subseteq
  T$ is an {\em inductive validity core} for $(I, T)\vdash P$ iff $(I,
  S) \vdash P$. When $I$, $T$, and $P$ can be inferred from
  context we will simply say $S$ is an inductive validity core.
\end{definition}

\begin{definition}{\emph{Minimal Inductive Validity Core:}}
  \label{def:minimal-ivc}
  An inductive validity core $S$ for $(I, T)\vdash P$ is minimal iff
  there does not exist $M \subset S$ such that $M$ is an inductive validity core
  for $(I, T)\vdash P$.
\end{definition}

Note that minimal inductive validity cores are not necessarily unique.
For example, take $I = a \land b$, $T = a' \land b'$, and $P = a \lor
b$. Then both $\{a'\}$ and $\{b'\}$ are minimal inductive validity
cores for $(I, T)\vdash P$. However, inductive validity cores do have
the following monotonicity property.

\begin{lemma}
  \label{lem:ivc-monotonic}
  Let $(I, T)$ be a transition system and let $P$ be a safety property
  with $(I, T)\vdash P$. Let $S_1 \subseteq S_2 \subseteq T$. If $S_1$
  is an inductive validity core for $(I, T)\vdash P$ then $S_2$ is an
  inductive validity core for $(I, T)\vdash P$.
\end{lemma}
\begin{proof}
  From $S_1 \subseteq S_2$ we have $S_2 \Rightarrow S_1$. Thus the
  reachable states of $(I, S_2)$ are a subset of the reachable states
  of $(I, S_1)$. \qed
\end{proof}

\begin{algorithm}[t]
  \SetKwInOut{Input}{input}
  \SetKwInOut{Output}{output}
  \Input{$(I, T)\vdash P$}
  \Output{Minimal inductive validity core for $(I, T)\vdash P$}
  \BlankLine
  $S \leftarrow T$ \\
  \For{$x \in S$} {
    \If{$(I, S\setminus\{x\}) \vdash P$}{
      $S \leftarrow S\setminus \{x\}$
    }
  }
  \Return{S}
\caption{\bfalg: Brute-force algorithm for computing a minimal IVC}
\label{alg:naive}
\end{algorithm}

This lemma gives us a simple, brute-force algorithm for computing
a minimal inductive validity core, Algorithm \bfalg~(\ref{alg:naive}). The
resulting set of this algorithm is obviously an inductive validity
core for $(I, T)\vdash P$. The following lemma shows that it is also
minimal.

\begin{lemma}
  The result of Algorithm~\ref{alg:naive} is a minimal inductive validity core
  for $(I, T)\vdash P$.
\end{lemma}
\begin{proof}
  Let the result be $R$. Suppose towards contradiction that $R$ is not
  minimal. Then there is an inductive validity core $M$ with $M
  \subset R$. Take $x \in R\setminus M$. Since $x \in R$ it must be
  that during the algorithm $(I, S\setminus\{x\})\vdash P$ is not true
  for some set $S$ where $R \subseteq S$. We have $M \subset R
  \subseteq S$ and $x\not\in M$, thus $M \subseteq S\setminus \{x\}$.
  Since $M$ is an inductive validity core,
  Lemma~\ref{lem:ivc-monotonic} says that $S\setminus \{x\}$ is an
  inductive validity core, and so $(I, S\setminus\{x\})\vdash P$. This
  is a contradiction, thus $R$ must be minimal.
\end{proof}

This algorithm has two problems. First, checking if a safety property
holds is undecidable in general thus the algorithm may never terminate
even when the safety property is easily provable over the original
transition system. Second, this algorithm is very inefficient since it
tries to re-prove the property multiple times.

\begin{algorithm}[t]
  \SetKwInOut{Input}{input}
  \SetKwInOut{Output}{output}
  \Input{$P$ with invariants $Q$ is $k$-inductive for $(I, T)$}
  \Output{Inductive validity core for $(I, T)\vdash P$}
  \BlankLine
  $k \leftarrow \mink(T, P \land Q)$ \\
  $R \leftarrow \reduceinv_k(T, Q, P)$ \\
  \Return{$\minivc_k(I, T, R)$}\\
\caption{\ucalg: Efficient algorithm for computing a nearly minimal inductive validity core from UNSAT cores}
\label{alg:ivc}
\end{algorithm}

The key to a more efficient algorithm is to make better use of the
information that comes out of model checking. In addition to knowing
that $P$ holds on a system $(I, T)$, suppose we also know something
stronger: $P$ with the invariant set $Q$ is $k$-inductive for $(I,
T)$. This gives us the broad structure of a proof for $P$ which allows
us to reconstruct the proof over a modified transition system.
However, we must be careful since this proof structure may be more
than is actually needed to establish $P$. In particular, $Q$ may
contain unneeded invariants which could cause the inductive validity
core for $P \land Q$ to be larger than the inductive validity core for
$P$. Thus before computing the inductive validity core we first try to
reduce the set of invariants to be as small as possible. This
operation is expensive when $k$ is large so as a first step we
minimize $k$. This is the motivation behind
Algorithm \ucalg~(\ref{alg:ivc}).

\begin{figure}
\begin{align*}
  &\bq_1(I, T, P) \equiv \forall s_0.~ I(s_0) \Rightarrow P(s_0) \\
  &\bq_{k+1}(I, T, P) \equiv \bq_k(I, T, P) \land~ \\
  &\hspace{10pt}\left(\forall s_0, \ldots, s_k.~ I(s_0) \land T(s_0,
  s_1) \land \cdots \land T(s_{k-1}, s_k) \Rightarrow P(s_k)\right)
  \\[5pt]
  &\iq_k(T, Q, P) \equiv (\forall s_0, \ldots, s_k.~\\
  &\hspace{10pt} Q(s_0) \land T(s_0,
  s_1) \land \cdots \land Q(s_{k-1}) \land T(s_{k-1}, s_k) \Rightarrow
  P(s_k)) \\[5pt]
  &\fq_k(I, T, P) \equiv \\
  &\hspace{10pt}\bq_k(I, T, P) \land \iq_k(T, P, P)
\end{align*}
\caption{$k$-induction queries}
\label{fig:queries}
\end{figure}

To describe the details of Algorithm~\ref{alg:ivc} we define queries
for the base and inductive steps of $k$-induction
(Figure~\ref{fig:queries}). Note, in $\iq(T, Q, P)$ we separate the
assumptions made on each step, $Q$, from the property we try to show
on the last step, $P$. We use this separation when reducing the set of
invariants.

We assume that our queries are checked by an SMT solver. That is, we
assume we have a function $\checksat(F)$ which determines if $F$, an
existentially quantified formula, is satisfiable or not. In order to
efficiently manipulate our queries, we assume the ability to create
{\em activation literals} which are simply distinguished Boolean
variables. The call $\checksat(A, F)$ holds the activation literals in
$A$ true while checking $F$. When $F$ is unsatisfiable, we assume we
have a function $\unsatcore()$ which returns a minimal subset of the
activation literals such that the formula is unsatisfiable with those
activation literals held true. In practice, SMT solvers often return a
non-minimal set, but we can minimize the set via repeated calls to
\checksat. We assume both \checksat and \unsatcore are always
terminating.

\begin{algorithm}[t]
  $k' \leftarrow 1$ \\
  \While{$\checksat(\neg\iq_{k'}(T, P, P)) = \sat$} {
    $k' \leftarrow k' + 1$ \\
    }
  \Return{$k'$} \\
\caption{$\mink(T, P)$}
\label{alg:minimize-k}
\end{algorithm}

The function $\mink(T, P)$ is defined in
Algorithm~\ref{alg:minimize-k}. This function assumes that $P$ is
$k$-inductive for $(I, T)$. It returns the smallest $k'$ such that $P$
is $k'$-inductive for $(I, T)$. We start checking at $k' = 1$ since
smaller values of $k'$ are much quicker to check than larger ones. The
checking must eventually terminate since $P$ is $k$-inductive. We also
only check the inductive query since we know the base query will be
true for all $k' \leq k$. Although we describe each query in
Algorithm~\ref{alg:minimize-k} separately, in practice they can be
done incrementally to improve efficiency.

\begin{algorithm}[t]
  $R \leftarrow \{P\}$ \\
  Create activation literals $A = \{a_1, \ldots, a_n\}$ \\
  $C \leftarrow (a_1 \Rightarrow Q_1) \land \cdots \land (a_n \Rightarrow Q_n)$ \\
  \While{$true$} {
    $\checksat(A, \neg\iq_k(T, C, R))$ \\
    \If{$\unsatcore() = \emptyset$}{
      \Return{R}
    }
    \For{$a_i \in \unsatcore()$}{
      $R \leftarrow R \cup \{Q_i\}$ \\
      $C \leftarrow C \setminus \{a_i \Rightarrow Q_i\}$ \\
    }
  }
\caption{$\reduceinv_k(T, \{Q_1, \ldots, Q_n\}, P)$}
\label{alg:reduce-invariants}
\end{algorithm}

The function $\reduceinv_k(T, \{Q_1, \ldots, Q_n\}, P)$ is defined in
Algorithm~\ref{alg:reduce-invariants}. This function assumes that $P
\land Q_1 \land \cdots \land Q_n$ is $k$-inductive for $(I, T)$. It
returns a set $R \subseteq \{P, Q_1, \ldots, Q_n\}$ such that $R$ is
$k$-inductive for $(I, T)$ and $P \in R$. Like \mink, this function
only checks the inductive query since each element of $R$ is an
invariant and therefore will always pass the base query. A significant
complication for reducing invariants is that some invariants may
mutually need each other, even though none of them are needed to prove
$P$. Thus in Algorithm~\ref{alg:reduce-invariants} we find a minimal
set of invariants needed to prove $P$, then we find a minimal set of
invariants to prove those invariants, and so on. We terminate when no
more invariants are needed to prove the properties in $R$.
Algorithm~\ref{alg:reduce-invariants} is guaranteed to terminate since
$R$ gets larger in every iteration of the outer loop and it is bounded
above by $\{P, Q_1, \ldots, Q_n\}$. As with
Algorithm~\ref{alg:minimize-k}, we describe each query in
Algorithm~\ref{alg:reduce-invariants} separately, though in practice
large parts of the queries can be re-used to improve efficiency.

This iterative lemma determination does not guarantee a minimal
result. For example, we may find $P$ requires just $Q_1$, that $Q_1$
requires just $Q_2$, and that $Q_2$ does not require any other
invariants. This gives the result $\{P, Q_1, Q_2\}$, but it may be
that $Q_2$ alone is enough to prove $P$ thus the original result is
not minimal. Also note, we do not care about the result of \checksat,
only the \unsatcore that comes out of it. Since $P \land Q_1 \land
\cdots \land Q_n$ is $k$-inductive, we know the \checksat call will
always return \unsat.

\begin{algorithm}[t]
  Create activation literals $A = \{a_1, \ldots, a_n\}$ \\
  $T \leftarrow (a_1 \Rightarrow T_1) \land \cdots \land (a_n \Rightarrow T_n)$ \\
  $\checksat(A, \neg\fq_k(I, T, P))$ \\
  $R \leftarrow \emptyset$ \\
  \For{$a_i \in \unsatcore()$}{
    $R \leftarrow R \cup \{T_i\}$
  }
  \Return{R}
\caption{$\minivc_k(I, \{T_1, \ldots, T_n\}, P)$}
\label{alg:minimize-ivc}
\end{algorithm}

The function $\minivc_k(I, \{T_1, \ldots, T_n\}, P)$ is defined in
Algorithm~\ref{alg:minimize-ivc}. This function assumes that $P$ is
$k$-inductive for $(I, T)$. It returns a minimal inductive validity
core $R \subseteq \{T_1, \ldots, T_n\}$ such that $P$ is $k$-inductive
for $(I, R)$. It is trivially terminating. Since
Algorithms~\ref{alg:minimize-k}, \ref{alg:reduce-invariants}, and
\ref{alg:minimize-ivc} are terminating, Algorithm~\ref{alg:ivc} is
always terminating.

Our full inductive validity core algorithm in Algorithm~\ref{alg:ivc}
does not guarantee a minimal inductive validity core. One reason is
that \reduceinv does not guarantee a minimal set of invariants. A
larger reason is that we only consider the invariants that the
algorithm is given at the outset. It is possible that there are other
invariants which could lead to a smaller inductive validity core, but
we do not search for them. In Sections~\ref{sec:experiment} and
\ref{sec:results}, we show that in practice our algorithm is nearly
minimal and much more efficient than the naive algorithm. The
following theorem shows that minimality checking is at least as hard
as model checking and therefore undecidable in many settings.

\begin{theorem}
\label{thm:minimal-hard}
Determining if an IVC is minimal is as hard as model checking.
\end{theorem}
\begin{proof}
Consider an arbitrary model checking problem $(I, T)\vdash^? P$ where
$P$ is not a tautology. We will construct an IVC for a related model
checking problem which will be minimal if and only if $(I, T)\nvdash
P$. Let $x$ and $y$ be fresh variables. Construct a transition system
with initial predicate $I\land \neg x$ and transition predicate $(x'
\Rightarrow y') \land ((y' \Rightarrow P') \land T)$. The constructed
system clearly satisfies the property $x \Rightarrow P$. Thus $S = \{x'
\Rightarrow y', (y' \Rightarrow P') \land T\}$ is an IVC. $S$ is
minimal if and only if neither $\{x' \Rightarrow y'\}$ nor $\{(y'
\Rightarrow P') \land T\}$ is an IVC. Since $x$ and $y$ are fresh and
$P$ is not a tautology, $\{x' \Rightarrow y'\}$ is not an IVC. Since
$x$ and $y$ are fresh, $\{(y' \Rightarrow P') \land T\}$ is an IVC for
the property $x \Rightarrow P$ if and only if $(I, T)\vdash P$.
Therefore, $S$ is minimal if and only if $(I, T)\nvdash P$.
\end{proof}

When minimality is a necessity, we can combine \bfalg and \ucalg into
a single algorithm which aims to efficiently guarantee minimality. The
hybrid algorithm, \ucbfalg, consists of running \ucalg to generate an
initial nearly minimal IVC which is then run through \bfalg to
guarantee minimality. The resulting algorithm is not guaranteed to
terminate since \bfalg is not guaranteed to terminate.



\section{Implementation}
\label{sec:impl}

We have implemented the inductive validity core algorithms in the
previous section in two tools: {\em JKind}, which performs the \ucalg
algorithm, and {\em JSupport}, which can compute either the \bfalg or
the \ucbfalg algorithm (using JKind as a subprocess). Moreover, our
implementation of \ucbfalg uses an additional feature of JKind to
store and re-use discovered invariants between separate runs. This
reduces some of the cost of attempting to re-prove a property multiple
times. These tools operate over the Lustre
language~\cite{Halbwachs91:lustre}, which we briefly illustrate below.

\vspace{0.1in}

\subsection{Lustre and IVCs}

Lustre~\cite{Halbwachs91:lustre} is a synchronous dataflow language
used as an input language for various model checkers. The textual
models in Figures~\ref{fig:ex-before} and \ref{fig:ex-after} are
written in Lustre. We will use model in Figure~\ref{fig:ex-before} as
a running example in this section. For our purposes, a Lustre program
consists of 1) input variables, {\tt x} in the example, 2) output
variables, {\tt a}, {\tt b}, and {\tt y} in the example, and 3) an
equation for each output variable. A Lustre program runs over discrete
time steps. On each step, the input variables take on some values and
are used to compute values for the output variables on the same step.
In addition, equations may refer to the previous value of a variable
using the {\tt pre} operator. This operator is underspecified in the
first step, so the arrow operator, {\tt ->}, is used to guard the
{\tt pre} operator. In the first step the expression {\tt e1 -> e2}
evaluates to {\tt e1}, and it evaluates to {\tt e2} in all other steps.

We interpret a Lustre program as a model specification by considering
the behavior of the program under all possible input traces. Safety
properties over Lustre can then be expressed as Boolean expressions in
Lustre. A safety property holds if the corresponding expression is
always true for all input traces. For example, the property for
Figure~\ref{fig:ex-before} is {\tt y >= 0}, which is a valid property.

It is straightforward to translate this interpretation of Lustre into
the traditional initial and transition relations. We will show this by
continuing with the example in Figure~\ref{fig:ex-before}. First we
introduce a new Boolean variable $init$ into the state space to denote
when the system is in its initial state, the state of the system prior
to initialization. In the initial state, all other variables are
completely unconstrained which models the underspecification of the pre
operator during the first step. Then we define,
\begin{align*}
  &I((x, a, b, y, \mathit{init})) = \mathit{init} \\
  &T((x, a, b, y, \mathit{init}), (x', a', b', y', \mathit{init'})) = \\
  &\hspace{1.5cm} (a' = f(x', \ite{init}{0}{y})) \land~ \\
  &\hspace{1.5cm} (b' = \ite{a' \geq 0}{a'}{-a'}) \land~ \\
  &\hspace{1.5cm} (y' = b' + (\ite{init}{0}{y})) \land ~\\
  &\hspace{1.5cm} \neg\mathit{init'}
\end{align*}
Note that $f$ is unspecified in Figure~\ref{fig:ex-before} and so also
in $T$. In a real system, $f$ would be defined in the Lustre model and
expanded in $T$. A safety property such as {\tt y >= 0} is translated
into $\mathit{init} \lor (y \geq 0)$. Nested uses of arrow and pre
operators are handled by introducing new output variables for nested
expressions, though such details are unimportant for our purposes.

Each equation in the Lustre program is translated into a single
top-level conjunct in the transition relation. This is very convenient
as the IVC of a Lustre property can be reported in terms of the output
variables whose equations are part of the IVC. Equivalently, the
interpretation of an IVC for a Lustre property is that any output
variable that is not part of the IVC can be turned into an input
variable, its equation thrown away, while preserving the validity of
the property. Thus the granularity of the IVC analysis is determined
by the granularity of the Lustre equations and can be adjusted by
introducing auxiliary variables for subexpressions if desired.

\subsection{JKind}

JKind~\cite{jkind} is an infinite-state model checker for safety
properties. JKind proves safety properties using multiple cooperative
engines in parallel including $k$-induction~\cite{SheeranSS00},
property directed reachability~\cite{Een2011:PDR}, and template-based
lemma generation~\cite{Kahsai2011}. JKind accepts Lustre programs
written over the theory of linear integer and real arithmetic. In the
back-end, JKind uses an SMT solver such as Z3~\cite{DeMoura08:z3},
Yices~\cite{Dutertre06:yices}, MathSAT~\cite{Cimatti2013:MathSAT}, or
SMTInterpol~\cite{Christ2012:SMTInterpol}.

JKind works on multiple properties simultaneously. When a property is
proven and IVC generation is enabled, an additional parallel engine
executes Algorithm~\ref{alg:ivc} to generate a nearly minimal IVC.

JKind accepts an annotation on its input Lustre program indicating
which outputs variables to consider for IVC generation. Output
variables not mentioned in the annotation are implicitly included in
all IVCs. This allows the implementation to focus on the variables
important to the user and ignore, for example, administrative
equations. This is even more important for tools which generate Lustre
as they often create many such administrative equations which simply
wire together more interesting expressions.


\section{Experiment}
\label{sec:experiment}


We would like to investigate both the {\em efficiency} and {\em
  minimality} of our three algorithms: the naive brute-force
algorithm (\bfalg), the UNSAT core-based algorithm (\ucalg), and the
combined UNSAT core followed by brute-force minimization algorithm
(\ucbfalg). Efficiency is computed in terms of wall-clock time: how
much overhead does the IVC algorithm introduce? Minimality is
determined by the size of the IVC: cores with a smaller number of
variables are preferred to cores with a larger number of variables.
Finally, we are interested in the {\em diversity} of solutions: how
often do different tools/algorithms generate different minimal IVCs?

The use of JKind allows additional dimensions to our investigation: it supports two different inductive algorithms: $k$-induction and PDR, and a ``fastest'' mode, that runs both algorithms in parallel.  In addition, JKind supports multiple back-end SMT solvers including Z3~\cite{DeMoura08:z3}, Yices~\cite{Dutertre06:yices}, MathSAT~\cite{Cimatti2013:MathSAT}, and SMTInterpol~\cite{Christ2012:SMTInterpol}.  We would like to determine whether the choice of inductive algorithm affects the size of the IVC, whether different solvers are more or less efficient at producing IVCs, and whether running different solvers/algorithms leads to {\em diversity} of IVC solutions.

Therefore, we investigate the following research questions:
\begin{itemize}
    \item \textbf{RQ1:} How expensive is it to compute inductive validity cores using the \bfalg, \ucalg, and \ucbfalg algorithms?
    \item \textbf{RQ2:} How close to minimal are the IVC sets computed by \ucalg as opposed to the (guaranteed minimal) \ucbfalg?  How do the sizes of IVCs compare to static slices of the model?
    \item \textbf{RQ3:} How much {\em diversity} exists in the solutions produced by different solver/induction algorithm configurations?
\end{itemize}

\subsection{Experimental Setup}
In this study, we started from a suite of 700 Lustre models developed
as a benchmark suite for~\cite{Hagen08:FMCAD}. We augmented this suite
with 81 additional models from recent verification projects including
avionics and medical devices~\cite{QFCS15:backes,hilt2013}. Most of
the benchmark models from~\cite{Hagen08:FMCAD} are small (10kB or less,
with 6-40 equations) and contain a range of hardware benchmarks and
software problems involving counters. The additional models are much
larger: around 80kB with over 300 equations. We added the new
benchmarks to better check the scalability for the tools, especially
with respect to the brute force algorithm.
%
Each benchmark model has a single property to analyze.  For our purposes, we are only interested in models with a {\em valid} property (though it is perhaps worth noting that there is no additional computation---and thus no overhead---using the JKind IVC options for {\em invalid} properties).  In our benchmark set, 295 models yield counterexamples, and 10 additional models are neither provable nor yield counterexamples in our test configuration (see next paragraph for configuration information).  The benchmark suite therefore contains 476 models with valid properties, which we use as our test subjects.

For each test model, we computed \ucalg in 12+1 configurations: the
twelve configurations were the cross product of all solvers \{Z3,
Yices, MathSAT, SMTInterpol\} and inductive algorithms
\{$k$-induction, PDR, fastest\}, and the remaining (+1) configuration
was an instance of \bfalg run on Yices, which is the default solver in
JKind. In addition, for each of the 12 configurations, we ran an
instance of JKind without IVC to examine overhead. The experiments
were run on an Intel(R) i5-2430M, 2.40GHz, 4GB memory machine, with a
1 hour timeout for each analysis on any model. The data gathered for
each configuration of each model included the time required to check
the model without IVC, with IVC, and also the set of elements in the
computed IVC.\footnote{The benchmarks, all raw experimental results,
  and computed data are available on \cite{expr}.}

Note that not all analysis problems were solvable with all algorithms: for all solvers, $k$-induction (without IVC) was unable to solve 172 of the examples.  When comparing minimality of different solving algorithms, we only considered cases where both algorithms provided a solution (as will be discussed in more detail in Section~\ref{sec:minimality}).

\section{Results}
\label{sec:results}

\newcommand{\takeaway}[1]{
\vspace{6pt}
\noindent\fbox{\parbox{0.975\columnwidth}{#1}}
\vspace{6pt}
}

In this section, we examine our experimental results from three perspectives: performance, minimality of \ucalg results, and diversity.

\subsection{Performance}
\label{sec:performance}

In this subsection, we examine the performance of our inductive validity core algorithms (research question \textbf{RQ1}).  First we examine the performance overhead of the \ucalg algorithm over the time necessary to find a proof using inductive model checking.  To examine this question, we use the default {\em fastest} option of JKind which terminates when either the $k$-induction or PDR algorithm finds a proof.  To measure the performance overhead of the \ucalg algorithm, we execute it over the proof generated by the {\em fastest} option.

Since the \ucalg algorithm uses the UNSAT core facilities of the
underlying SMT solver, the performance is dependent on the efficiency
of this part of the solver. Looking at Tables~\ref{tab:runtime-ucalg}
and~\ref{tab:overhead-ucalg}, it is possible to examine both the
computation time for analysis using the four solvers under evaluation
and the overhead imposed by the \ucalg algorithm.
Figure~\ref{fig:performance} allows a visualization of the runtime for
the \ucalg algorithm running different solvers. The data suggests that
Yices (the default solver in JKind) and Z3 are the most performant
solvers both in terms of computation time and overhead.





\takeaway{The \ucalg algorithm using the Z3 and Yices SMT solvers adds a modest performance penalty to the time required for inductive proofs.}

Next, we consider the overhead of \ucalg vs.\ \bfalg.  Recall from Section~\ref{sec:ivc} that \bfalg requires $n$ model checking runs, where $n$ is the number of conjuncts in the transition relation. As expected, the performance is approximately a linear multiple of the size of the model, so larger models yield substantially lower performance.\footnote{for Lustre models, the number of conjuncts is equivalent to the number of equations in the Lustre model.}  We run the brute-force algorithm using Yices as it is the default solver for JKind and is close to Z3 in terms of computation time.  For 19 models, \bfalg times out after 1 hour.   Figure~\ref{fig:runtimeall} shows the overhead of \bfalg in comparison to \ucalg with multiple solvers.

\takeaway{The brute-force algorithm \bfalg adds a substantial performance penalty to inductive proofs in all cases and is not scalable enough to compute a minimal core for large analysis problems.}

Finally, we consider the combined \ucbfalg algorithm, in which we
first run the \ucalg to determine a close-to-minimal IVC, then run
\bfalg on the remaining set. The overhead of this algorithm is
considered in Tables~\ref{tab:runtime-ucbfalg}
and~\ref{tab:overhead-ucbfalg}. While considerably slower than \ucalg,
this approach can still be used for reasonably sized models.

\begin{figure*}
  \centering
  \includegraphics[width=0.75\textwidth]{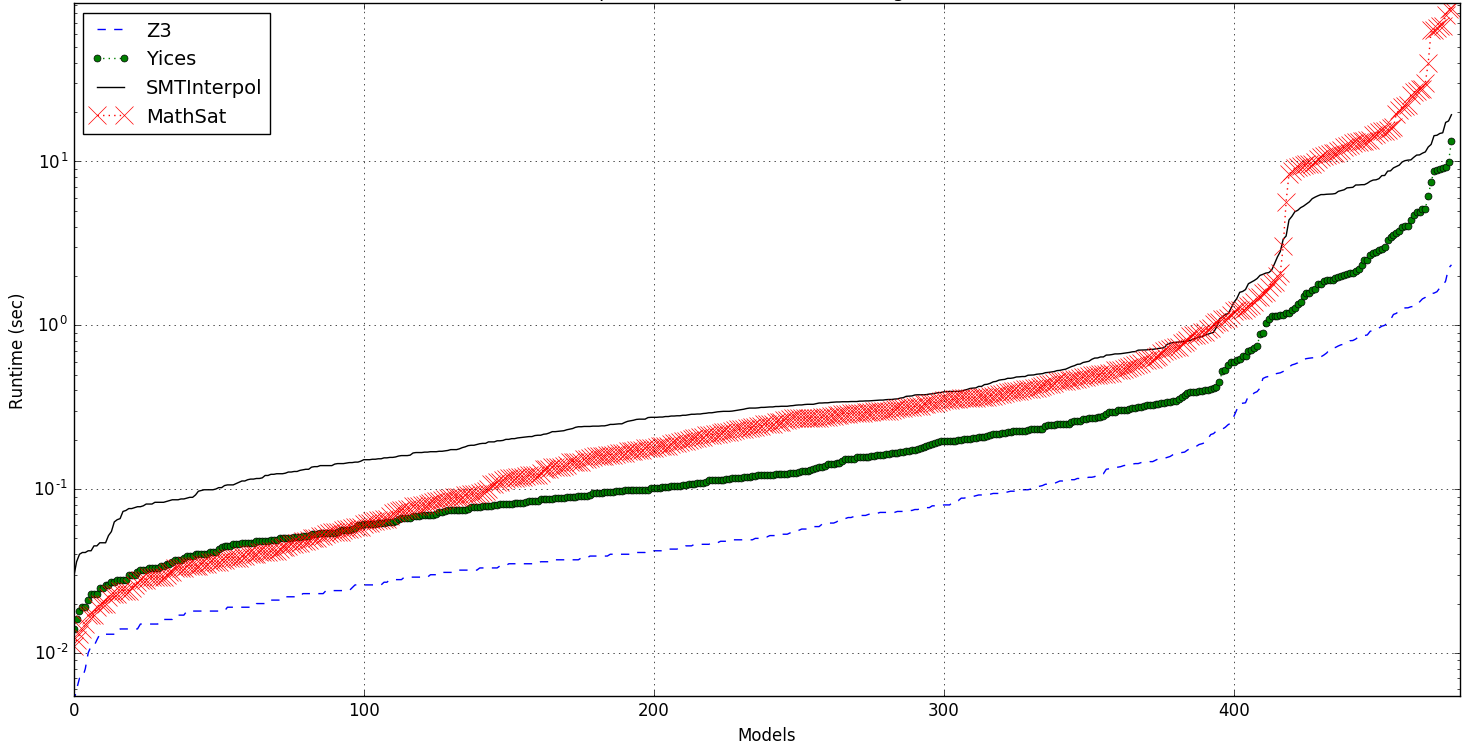}
  \label{fig:performance}
  \vspace{-0.1in}
  \caption{\ucalg performance on different solvers}
\end{figure*}

\begin{figure*}
  \centering
  \includegraphics[width=0.75\textwidth]{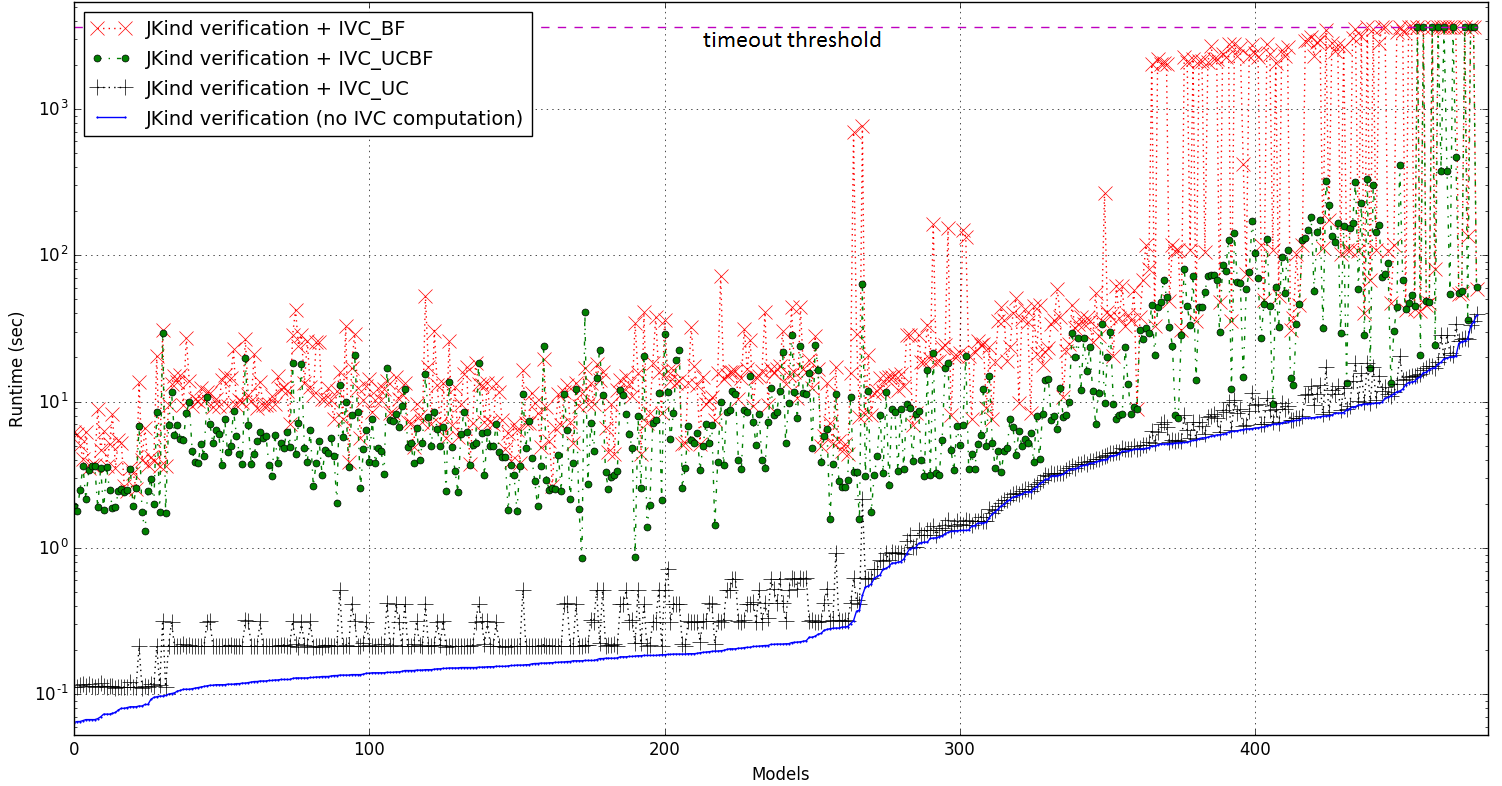}
  \vspace{-0.1in}
  \caption{Runtime of \bfalg, \ucbfalg, \ucalg algorithms for Yices}\label{fig:runtimeall}
\end{figure*}

\begin{table}
  \caption{\ucalg runtime with different solvers}
  \centering
  \begin{tabular}{ |c||c|c|c|c| }
    \hline
     runtime (sec) & min & max & mean & stdev \\[0.5ex]
    \hline\hline
    Z3   & 0.005 & 2.335 & 0.192 & 0.355 \\[0.5ex]
    Yices &   0.014  & 13.297   & 0.589 & 1.473 \\[0.5ex]
    SMTInterpol& 0.029 & 19.254 &  1.396 & 2.991 \\[0.5ex]
    MathSAT & 0.011 & 86.421 &  3.071 & 10.403 \\[0.5ex]
    \hline
  \end{tabular} \\
  \label{tab:runtime-ucalg}
\end{table}

\begin{table}
  \caption{Overhead of \ucalg computations using different solvers}
  \centering
  \begin{tabular}{ |c||c|c|c|c| }
    \hline
     solver & min & max & mean & stdev \\[0.5ex]
    \hline
    Z3   & 0.73\% & 84.13\% & 17.38\% & 16.92\% \\[0.5ex]
    Yices &   0.17\%  & 351.47\%   & 52.20\% & 54.50\% \\[0.5ex]
   SMTInterpol& 1.46\% & 175.75\% &  46.81\% & 37.35\%\\[0.5ex]
    MathSAT & 0.78\% & 955.52\% &  80.21\% & 112.92\%\\[0.5ex]
    \hline
  \end{tabular}
  \label{tab:overhead-ucalg}
\end{table}

\begin{table}
  \caption{\ucbfalg runtime}
  \centering
  \begin{tabular}{ |c||c|c|c|c| }
    \hline
     runtime (sec) & min & max & mean & stdev \\[0.5ex]
    \hline
    Yices &   0.68  & 3600.0   & 91.59 & 490.01 \\[0.5ex]
    \hline
    Z3 &   0.66  & 3600.0   & 93.01 & 490.27 \\[0.5ex]
    \hline
  \end{tabular}
  \label{tab:runtime-ucbfalg}
\end{table}

\begin{table}
  \caption{Overhead of \ucbfalg algorithm}
  \centering
  \begin{tabular}{ |c||c|c|c|c| }
    \hline
     solver & min & max & mean & stdev \\[0.5ex]
    \hline
    Yices & 122.50\%  & 30092.78\%   & 3195.90\% & 3896.05\% \\[0.5ex]
    \hline
    Z3 & 101.70\%  & 28114.07\%   & 3190.18\% & 4119.14\% \\[0.5ex]
    \hline
  \end{tabular}
  \label{tab:overhead-ucbfalg}
\end{table}

\begin{table}
  \caption{Aggregate IVC sizes produced by \ucalg\ using different inductive algorithms and solvers}
  \centering
  \begin{tabular}{ |c|c|c|c| }
    \hline
     solver & PDR & $k$-induction & \textbf{total} \\
    \hline
      Z3 & 2378 & 2379 & 4757 \\
      Yices & 2384 & 2376 & 4760 \\
      MathSAT & 2375 & 2369 & 4744 \\
      SMTInterpol & 2378 & 2368 & 4746 \\
    \hline
      \textbf{total} & 9515 & 9492 &   \\
    \hline
  \end{tabular}
  \label{tab:minimality-algorithm-solvers}
\end{table}

\begin{table}
  \caption{Increase in IVC Size for \ucalg\ vs.\ \ucbfalg}
  \centering
  \begin{tabular}{ |c||c|c|c|c| }
    \hline
     solver & min & max & mean & stdev \\[0.5ex]
    \hline
    Yices &   0.0\%   & 725.0\% & 20.54\% & 50.47\% \\[0.5ex]
    Z3 &   0.0\%   & 725.0\% & 20.81\% & 50.34\% \\[0.5ex]
    \hline
  \end{tabular}
  \label{tab:increase-ucalg-ucbfalg}
\end{table}

\begin{figure*}
  \centering
  \includegraphics[width=0.75\textwidth]{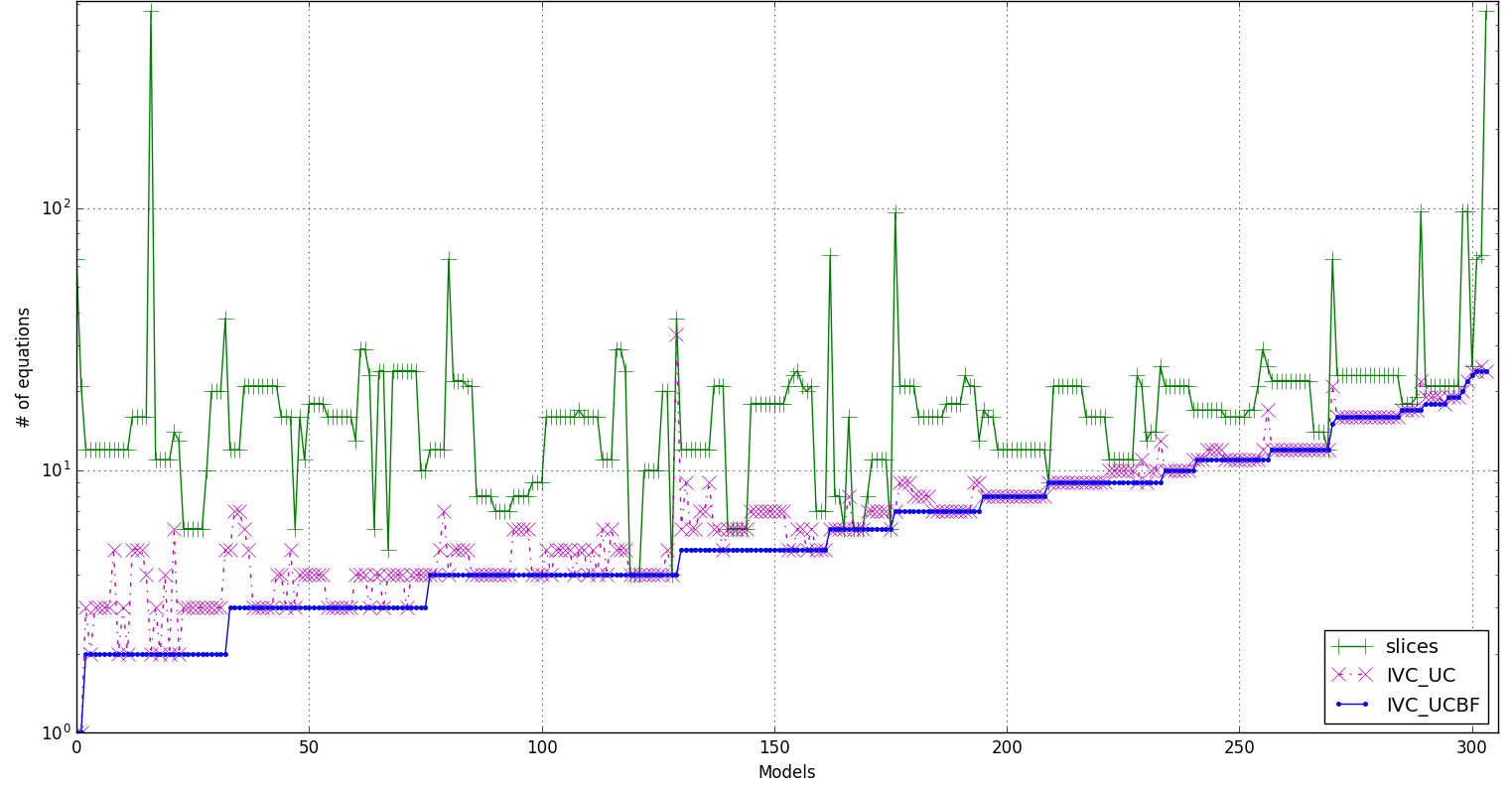}
  \vspace{-0.1in}
  \caption{IVC sizes produced by \ucalg, \ucbfalg for Yices vs. static slices}
  \label{fig:minimality-all}
\end{figure*}

\subsection{Minimality}
\label{sec:minimality}
In this section, we examine the minimality of the cores computed by
the \ucalg\ and \ucbfalg\ algorithms using different inductive proof methods, and we compare both algorithms against a {\em backward static slice}~\cite{Tip95asurvey} of the Lustre program starting with the property of interest.
There are three interesting aspects to be examined
related to this research question.  First (\textbf{RQ2.1}), does the
choice of SMT solver or algorithm used to produce a proof
($k$-induction or PDR) matter in terms of the minimality of the
inductive core?  As mentioned in Section~\ref{sec:ivc}, the \ucalg
algorithm is not guaranteed to produce a minimal core due in part to
the role of invariants used in producing a proof; as $k$-induction and
PDR use substantially different invariant generation algorithms, it is
likely that the set of necessary invariants for proofs are dissimilar,
and that this would in turn affect the number of model elements required for
the proof.  It is possible that one or the other algorithm is more likely
to yield smaller invariant sets.  In addition, differences in the choice of the
UNSAT core algorithms in the different solvers could affect the size of the
generated core. However, our algorithm already performs a minimization
step on UNSAT cores, and thus the only differences would be due to one
algorithm leading to a different minimal core than another.

As discussed in Section~\ref{sec:experiment}, $k$-induction is unable to solve all of the analysis problems; therefore we include only models that are solvable using {\em both} $k$-induction and PDR by {\em all solvers}, 304 models in all.  Examining the aggregate data in Table~\ref{tab:minimality-algorithm-solvers}, we can see the sizes of cores produced by different algorithms and solvers.




\takeaway{Neither PDR nor $k$-induction yields a smaller inductive
  validity core in general. The choice of underlying SMT solver does
  not substantially affect the size of the inductive validity cores.}


The next question (\textbf{RQ2.2}) asks how close to minimal are the
cores produced by \ucalg vs.\ the (guaranteed minimal) cores produced
by the \ucbfalg algorithm? Note that we cannot measure the distance on
all models because the combined algorithm times out on 9 of the larger
models. We therefore examine the distance from minimal cores produced
by the combined algorithm for models in which it completes within the
one hour timeout. For comparison, we run the \ucalg algorithm using Z3
and Yices with JKind's default {\em fastest} algorithm, which will use the
result of either $k$-induction or PDR. A graph showing the size of the
IVCs for each model produced using the Yices solver is shown in
Figure~\ref{fig:minimality-all}. In the figure, the models are ranked
along the x-axis by the size of the core produced by \ucbfalg. The
figure demonstrates that while on average there is a modest change in
minimality, there can be substantial variance on the sizes of the
cores produced by the \ucalg algorithm. Summary statistics are shown
in Table~\ref{tab:increase-ucalg-ucbfalg}.

\takeaway{The \ucalg algorithm computes cores that are on average 21\%
  larger than those produced by \ucbfalg, with substantial variance in
  some cases.}

The final question (\textbf{RQ2.3}) asks how well the approach compares to {\em backwards static slicing}~\cite{Tip95asurvey}, since slicing also reduces the set of model elements necessary to construct a proof.  We start the slice from the equation defining the property of interest, and use the usual approach~\cite{Gaucher03:slicing} that performs an iterative backward traversal from the variables used within an equation to their defining equations.  We expect the IVC mechanism to be more precise, because the slice overapproximates of the set of equations necessary for {\em any} proof.  This claim is demonstrated in Figure~\ref{fig:minimality-all}; slices are (mean) 406\% larger than the IVCs produced by our \ucalg algorithm and 465\% larger than those produced by \ucbfalg algorithm.

\takeaway{Both IVC algorithms compute cores that are usually much smaller than backwards static slices.}

Comparing the sizes of the \ucalg IVCs to the original models, the original 395 benchmark models from~\cite{Hagen08:FMCAD} already had applied slicing, so there is no difference between the sliced size and the original model size.  For the remaining 81 benchmarks, the number of equations is (mean) 2500\% larger than the \ucalg\ IVCs.  We note, however, that comparison of IVC size against the original model size can be misleading, as the improvement can easily be ``gamed'' by adding equations that are irrelevant to the property. 

\begin{table}
  \caption{Pairwise Jaccard distances among all models}
  \centering
  \begin{tabular}{ |c|c|c|c| }
    \hline
     min & max & mean & stdev \\[0.5ex]
    \hline
     0.0   & 0.878 & 0.026 & 0.059 \\[0.5ex]
    \hline
  \end{tabular}
  \label{tab:jaccard-avg}
\end{table}

\begin{figure*}
  \centering
  \vspace{3mm}
  \includegraphics[width=0.75\textwidth]{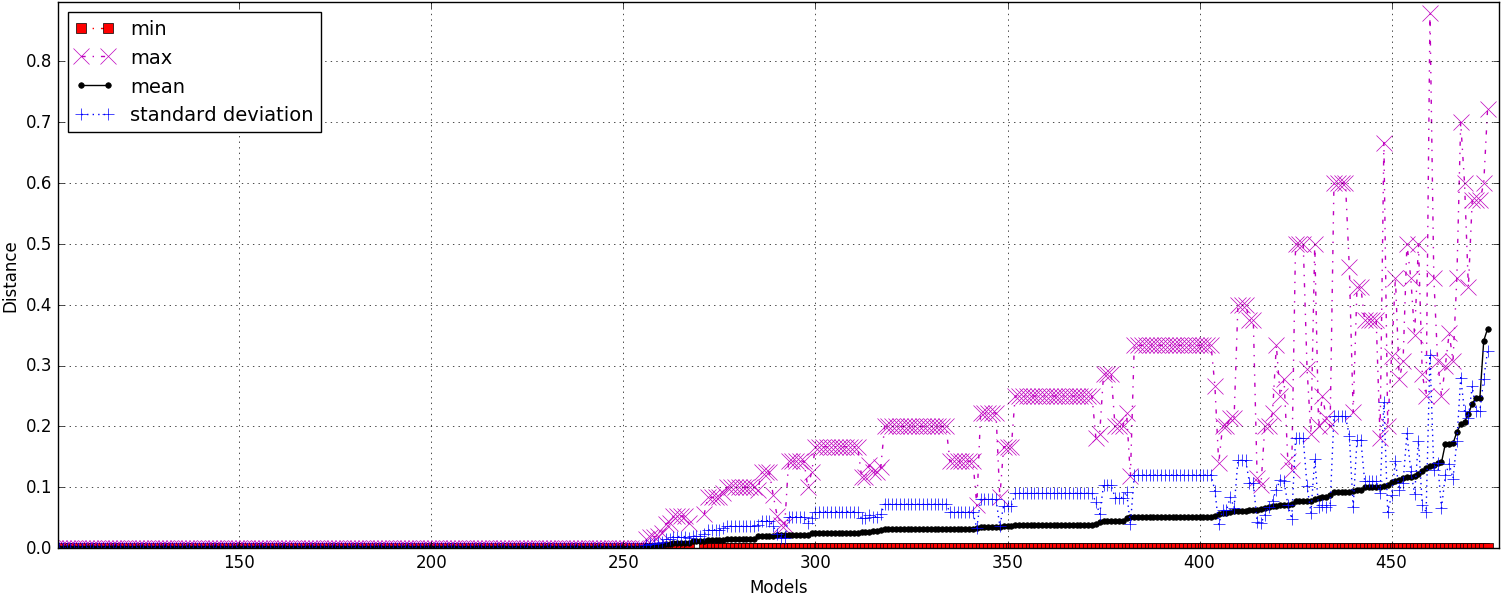} \\
  \vspace{-0.1in}
  \caption{Pairwise Jaccard distance between IVCs}\label{fig:jacdis}
\end{figure*}

\subsection{Diversity}
\label{sec:diversity}



Recall from Section~\ref{sec:ivc} that a {\em minimal}
IVC set is any set leading to a proof such that if you remove any of
its elements, it no longer produces a proof. For certain models and
properties, it is possible that there are many minimal cores that will
lead to a proof. In this section, we examine the issue of diversity:
do different solvers and algorithms lead to {\em different} minimal
cores? This is both a function of the models and the solution
algorithms: for certain models, there is only one possible minimal IVC
set, whereas other models might have many. Given that there are
multiple solutions, the interesting question is whether using
different solvers and algorithms will lead to different solutions.
The reason diversity is considered is that it has substantial relevance to
some of the uses of the tool, e.g., for constructing multiple traceability
matrices from proofs (see Section~\ref{sec:conc}).
Note that our exploration in this experiment is not
exhaustive, but only exploratory, based on the IVCs returned by different
algorithms and tools; we leave exhaustive exploration of
IVCs for future work.


To measure diversity of IVCs, we use Jaccard distance:
\begin{definition}{\emph{Jaccard distance:}}
  \label{def:dj}
  $d_J(\small{A}, \small{B}) = 1 - \frac{|A \cap B|}{|A \cup B|} ,\\ 0 \leq d_J(\small{A}, \small{B}) \leq 1$
\end{definition}
\noindent Jaccard distance is a standard metric for comparing finite
sets (assuming that both sets are non-empty) by comparing the size of
the intersection of two sets over its union. For each model in the
benchmark, the experiments generated 13 potentially different IVCs. Therefore, we
obtained $\binom{13}{2} = 78$ combinations of pairwise distances per
model. Then, minimum, maximum, average, and standard deviation of the
distances were calculated (Figure~\ref{fig:jacdis}), by which, again,
we calculated these four measures among all models. As seen in
Table~\ref{tab:jaccard-avg}, on average, the Jaccard distance between
different solutions is small, but the maximum is close to 1, which
indicates that even for our exploratory analysis, there are models for
which the tools yield substantially diverse solutions. The diversity
between solutions is represented graphically in
Figure~\ref{fig:jacdis}, where for each model, we present the min,
max, and mean pairwise Jaccard distance of the solutions produced by algorithm
\ucalg for each model, ranked by the mean distance.

\subsection{Discussion}

In the previous section, we presented three algorithms for determining
inductive validity cores. The brute-force algorithm is guaranteed
minimal, but is often very slow. The other two algorithms, the UNSAT
core algorithm \ucalg and the combined algorithm \ucbfalg, represent
interesting trade-offs. The \ucalg algorithm is much faster, but is
not guaranteed to be minimal; the result of this algorithm can be
further, and sometimes quickly, refined by the combined algorithm.
Thus, we can choose to trade off speed for guaranteed minimality using
these two algorithms; the combined algorithm can be viewed as a
refinement algorithm that we can terminate either at completion or
after a fixed time bound.


Although our experiment does not ask statistical questions, it is still worth examining threats towards generalizing our results.  First, are the models and properties that we chose representative?  We started from an existing benchmark from another research group suite to try to assuage this concern, but most of these models were small, so we extended the benchmark suite with 81 of our own models.  It is possible that our additions skew the results, though these models are immediately derived from previously published work and not modified for our analysis here.  Second, our models and tools use the Lustre language, which is equational, rather than conjuncted transition systems; it is possible (though, in our opinion, unlikely) that arbitrary conjuncts rather than equations will yield different performance or minimality characteristics.

Our approach is limited by the capabilities of the SMT solvers and inductive model checking algorithms that are used.  For example, it is difficult, given state of the art SMT solvers, to produce proofs involving complex models involving non-linear floating-point arithmetic.  However, given an inductive proof produced by an UNSAT-core-producing SMT solver, we feel confident that the \ucalg\ algorithm can produce an IVC.  Our approach is theory and invariant-generator agnostic, so as inductive model checking algorithms evolve and SMT solvers add support for new theories, the IVC algorithm should be able to work without modification.


\section{Related work}
\label{sec:related}


Our work builds on top of a substantial foundation building Minimally Unsatisfiable Subformulas
(MUSes) from UNSAT cores~\cite{Cimatti2007:UNSAT}, including \cite{marques2010minimal, belov2012towards, ryvchin2011faster, belov2012computing, nadel2010boosting}.  Recent algorithms can handle very large problems, but computing MUSes is still a resource-intensive task.  While some work is aimed at providing a set of minimal unsatisfiable formulae, minimality is usually defined such that given a set of clauses $\mathbb{M}$, removing any member of $\mathbb{M}$ makes it satisfiable \cite{belov2012computing}.  The step of producing minimal invariants for proofs has been investigated in depth by Ivrii et al. in~\cite{Ivrii14:invariants}.

UNSAT cores and MUSes are used for many different activities within
formal verification. Gupta et al. \cite{gupta2003iterative} and
McMillan and Amla \cite{mcmillan2003automatic} introduced the use of
unsatisfiable cores in proof-based abstraction engines. Their goal is
to shrink the abstraction size by omitting the parts of the design
that are irrelevant to the proof of the property under verification.
Torlak et al. in~\cite{Torlak08:cores} finds MUSes of Alloy
specifications, and considers semantic vacuity, which we consider in
Section~\ref{sec:intro}. Alloy models are only analyzed up to certain
size bounds, however, and in general are unable to prove properties
for arbitrary models. Also, because we are extracting information from
proofs, it is possible to use IVCs for additional purposes (proof
explanation and completeness checking).

If we view Lustre as a programming language, our work can be viewed as a more accurate form of program slicing~\cite{Tip95asurvey}.  We perform {\em backwards slicing} from the formula that defines the property of interest of the model.  The slice produced is smaller and more accurate than a static slice of the formula~\cite{Weiser:1981:slicing}, but guaranteed to be a sound slice for the formula for all program executions, unlike dynamic slicing~\cite{Agrawal:1990:slicing}.  Predicate-based slicing has been used~\cite{Li04:slicing} to try to minimize the size of a dynamic slice.  Our approach may have utility for some concerns of program slicing (such as model understanding) by constructing simple ``requirements'' of a model and using the tool to find the relevant portions of the model.

Another potential use of our work is for ``semantic'' vacuity detection.  A standard definition of vacuity is syntactic and defined as follows~\cite{Kupferman:2006:SCF}: {\em A system K satisfies a formula $\phi$ vacuously iff $K \vdash \phi$ and there is some subformula $\psi$ of $\phi$ such that $\psi$ does not affect $\phi$ in K}.  Vacuity has been extensively studied~\cite{Gurfinkel:2012:RVB,Chockler2008,DBLP:Ben-DavidK13,Kupferman:2006:SCF,Chockler:2007,Beer1997} considering a range of different temporal logics and definitions of ``affect''.  On the other hand, our work can be used to consider a broader definition of vacuity.  Even if all subformulae are required (the property is not syntactically vacuous), it may not require substantial portions of the model, and so may be provable for vacuous reasons.  The problem is exacerbated when the modeling and property language are the same (as in JKind), because whether a subformula is considered part of the model or part of the property, from the perspective of checking tools, can be unclear.

Determining completeness of properties has also been extensively studied. Certification standards such as DO-178C~\cite{DO178C} require that requirements-derived tests achieve some level of structural coverage (MC/DC, decision, statement) depending on the criticality level of the software, in order to approximate completeness.  If coverage is not achieved, then additional requirements and tests are added until coverage is achieved.  Chockler~\cite{chockler_coverage_2003} defined the first completeness metrics directly on formal properties based on mutation coverage.  Later work by Kupferman et al.~\cite{Kupferman:2006:SCF} defines completeness as an extension of vacuity to elements in the model.  We present an alternative approach that uses the proof directly, which we expect to be considerably less expensive to compute.  Recent work by Murugesan~\cite{murugesan2015we} and Schuller~\cite{schuler_assessing_2011} attempts to combine test coverage metrics with requirements to determine completeness.

\section{Conclusions \& Future Work}
\label{sec:conc}

In this paper, we have defined the notion of inductive validity core (IVC) which
appears to be a useful measure in relation to a valid safety property
for inductive model checking. We have presented a novel algorithm for
computing IVCs that are nearly minimal and have shown that full
minimality is undecidable in many settings. Our algorithm is
applicable to all forms of inductive SAT/SMT-based model checking
including $k$-induction, PDR, and
interpolation-based model checking.
We have implemented our IVC algorithm as part of the open source model
checker JKind. We have shown that the algorithm requires only a
moderate overhead and produces nearly minimal IVCs in practice.
Moreover, the produced IVCs are fairly stable with respect to
underlying proof engines ($k$-induction and PDR) and back-end SMT
solvers (Yices, Z3, MathSAT, SMTInterpol).

Our work has recently been integrated into the AADL/AGREE tool
suite~\cite{QFCS15:backes,hilt2013}, which supports compositional
reasoning about system architectures.  First, IVCs are used to
to automatically compute traceability information between high- and
low-level requirements in compositional proofs. Second, IVCs are
used by the AGREE symbolic simulator to explain conflicts when the
simulator is not able to compute a ``next state'' for a set of chosen
constraints.  A pilot project at Rockwell
Collins is using the traceability information produced by the IVC
support in the AGREE tool.

In future work, we will compare the traceability matrices
generated by IVCs with those produced by human experts and and by
automated heuristic approaches.  Our expectation is that the traceability
information produced by IVCs will be both more accurate and closer to
minimal than other approaches.
We also will examine the impact of multiple distinct IVCs on traceability
research.  An initial paper on this work, which we call {\em complete traceability}
has been accepted to the RE@Next! track of the Requirements Engineering
conference~\cite{Murugesan16:renext}.  We are interested in diversity both
in terms of regression analysis for testing and proof, as well as examining
the underlying sources of diversity in our analysis models.  We suspect that
in some cases, it indicates fault tolerance in the architecture under analysis,
and in other cases it may indicate redundancy in requirements specifications
for subcomponents.  To support a systematic investigation of diversity, we
plan to investigate algorithms for exploring the space of IVCs, e.g., finding
a minimum, rather than minimal IVC, or finding all IVCs.

Finally, we are in the process of
comparing our approach against other approaches measuring completeness of
requirements (such as those in~\cite{chockler_coverage_2003, Kupferman:2006:SCF, kupferman_theory_2008}).
%



\textbf{Acknowledgments:}
This work was supported by DARPA under contract FA8750-12-9-0179 (Secure Mathematically-Assured Composition of Control Models) and by NASA under contract NNA13AA21C (Compositional Verification of Flight Critical Systems).


\bibliographystyle{abbrv}
\bibliography{biblio}

~

%



\end{document}